\documentclass[11pt,letterpaper,onecolumn,accepted=2021-09-02]{quantumarticle}%
\pdfoutput=1 

\usepackage[utf8]{inputenc}
\usepackage[english]{babel}
\usepackage[T1]{fontenc}

\DeclareMathAlphabet{\mathbbold}{U}{bbold}{m}{n}

\usepackage[numbers,sort&compress]{natbib}
\usepackage{amsmath,amsfonts,amssymb}
\usepackage{amsthm}
\usepackage{graphicx}
\usepackage{fullpage}
\usepackage[pagebackref]{hyperref}
\usepackage{thm-restate}
\usepackage[capitalize,nameinlink]{cleveref}

\usepackage{tikz}
\usepackage{pgfplots}
\pgfplotsset{compat=1.13}
\usepackage{enumerate}
\usepackage{qcircuit}

\allowdisplaybreaks

\providecommand{\U}[1]{\protect\rule{.1in}{.1in}}
\newtheorem{theorem}{Theorem}

\newtheorem{definition}[theorem]{Definition}

\newtheorem{fact}[theorem]{Fact}
\newtheorem{lemma}[theorem]{Lemma}

\newtheorem{problem}[theorem]{Problem}
\newtheorem{proposition}[theorem]{Proposition}

\newcommand{\id}{\mathbbold{1}}
\newcommand{\E}{\mathop{\mathbb{E}}}
\newcommand{\Var}{\mathop{\mathrm{Var}}}

\newcommand{\eps}{\varepsilon}
\newcommand{\Reals}{\mathbb{R}}

\newcommand{\mathify}[1]{\ifmmode{#1}\else\mbox{$#1$}\fi}

\newcommand{\bra}[1]{\langle #1 |}
\newcommand{\ket}[1]{| #1 \rangle}
\newcommand{\braket}[1]{\langle #1 \rangle}

\newcommand{\Refl}{\mathcal{R}}
\newcommand{\Oracle}{\mathcal{O}}
\newcommand{\poly}{\mathrm{poly}}
\newcommand{\diamondnorm}[1]{\left|\left| #1 \right|\right|_\diamond}
\newcommand{\FS}{\textsc{Fourier Sampling}}
\newcommand{\Fn}{\mathcal{F}_n}
\newcommand{\Half}{\mathrm{Half}}

\begin{document}

\title{The Quantum Supremacy Tsirelson Inequality}

\author{William Kretschmer}
\email{kretsch@cs.utexas.edu}
\thanks{A preliminary version of this work appeared in the proceedings of the 12th Innovations in Theoretical Computer Science Conference (ITCS 2021) \cite{Kre21}.}
\affiliation{Department of Computer Science, The University of Texas at Austin, Austin, TX 78712, USA}

\maketitle

\begin{abstract}
A leading proposal for verifying near-term quantum supremacy experiments on noisy random quantum circuits is linear cross-entropy benchmarking. For a quantum circuit $C$ on $n$ qubits and a sample $z \in \{0,1\}^n$, the benchmark involves computing $|\langle z|C|0^n \rangle|^2$, i.e. the probability of measuring $z$ from the output distribution of $C$ on the all zeros input. Under a strong conjecture about the classical hardness of estimating output probabilities of quantum circuits, no polynomial-time classical algorithm given $C$ can output a string $z$ such that $|\langle z|C|0^n\rangle|^2$ is substantially larger than $\frac{1}{2^n}$ (Aaronson and Gunn, 2019). On the other hand, for a random quantum circuit $C$, sampling $z$ from the output distribution of $C$ achieves $|\langle z|C|0^n\rangle|^2 \approx \frac{2}{2^n}$ on average (Arute et al., 2019).

In analogy with the Tsirelson inequality from quantum nonlocal correlations, we ask: can a polynomial-time quantum algorithm do substantially better than $\frac{2}{2^n}$? We study this question in the query (or black box) model, where the quantum algorithm is given oracle access to $C$. We show that, for any $\varepsilon \ge \frac{1}{\mathrm{poly}(n)}$, outputting a sample $z$ such that $|\langle z|C|0^n\rangle|^2 \ge \frac{2 + \varepsilon}{2^n}$ on average requires at least $\Omega\left(\frac{2^{n/4}}{\mathrm{poly}(n)}\right)$ queries to $C$, but not more than $O\left(2^{n/3}\right)$ queries to $C$, if $C$ is either a Haar-random $n$-qubit unitary, or a canonical state preparation oracle for a Haar-random $n$-qubit state. We also show that when $C$ samples from the Fourier distribution of a random Boolean function, the naive algorithm that samples from $C$ is the optimal 1-query algorithm for maximizing $|\langle z|C|0^n\rangle|^2$ on average.
\end{abstract}




\section{Introduction}
A team based at Google has claimed the first experimental demonstration of quantum computational supremacy on a programmable device \cite{AAB+19}. The experiment involved \textit{random circuit sampling}, where the task is to sample (with nontrivial fidelity) from the output distribution of a quantum circuit containing random $1$- and $2$-qubit gates. To verify their experiment, they used the so-called \textit{Linear Cross-Entropy Benchmark}, or Linear XEB. Specifically, for an $n$-qubit quantum circuit $C$ and samples $z_1,\ldots,z_k \in \{0,1\}^n$, the benchmark is given by:
$$b = \frac{2^n}{k} \cdot \sum_{i=1}^k |\braket{z_i|C|0^n}|^2.$$
The goal is for $b$ to be large with high probability over the choice of the random circuit and the randomness of the sampler, as this demonstrates that the observations tend to concentrate on the outputs that are more likely to be measured under the ideal distribution for $C$ (i.e. the noiseless distribution in which $z$ is measured with probability $|\braket{z|C|0^n}|^2$). We formalize this task as the $b$-XHOG task:

\begin{problem}[$b$-XHOG, or Linear Cross-Entropy Heavy Output Generation]
\label{prob:xhog}
Fix a distribution $\mathcal{D}$ over quantum circuits on $n$-qubits.\footnote{We will sometimes leave the choice of $\mathcal{D}$ implicit when it is clear from context.} Given a quantum circuit $C$ sampled from $\mathcal{D}$, output a sample $z \in \{0,1\}^n$ such that:
\[
\E_{C \sim \mathcal{D}}\left[|\braket{z|C|0^n}|^2\right] \ge \frac{b}{2^n}.
\]
\end{problem}

We emphasize that in this work, $b$-XHOG is fundamentally an average-case problem: it is always defined with respect to a distribution $\mathcal{D}$ over the choice of circuit $C$, which is why we require that the Linear XEB score exceeds $b$ \textit{in expectation}. Note also that the procedure for solving XHOG could itself be randomized (e.g. if it involves sampling $z$ from the output of a quantum device), so the expectation in the above definition is taken with respect to this randomness as well.


For the purposes of demonstrating quantum supremacy, we will typically think of ``large'' $b$ as any $b$ bounded away from $1$, as guessing $z$ uniformly at random achieves $b = 1$ (regardless of the distribution $\mathcal{D}$!), just because $\sum_{z \in \{0,1\}^n} |\braket{z|C|0^n}|^2 = 1$ for any $n$-qubit circuit $C$. On the other hand, sampling from the noiseless output distribution of $C$ achieves $b \approx 2$ when $\mathcal{D}$ selects random circuits of polynomial size and sufficient depth. Indeed, sampling from $C$ achieves $b \approx 2$ whenever the circuit distribution $\mathcal{D}$ empirically exhibits the \textit{Porter-Thomas} distribution on circuit output probabilities, in which the output probabilities are approximately i.i.d. exponential random variables \cite{AAB+19,AG19}. 

Under a strong complexity-theoretic conjecture about the classical hardness of nontrivially estimating output probabilities of quantum circuits, Aaronson and Gunn showed that no classical polynomial-time algorithm can solve $b$-XHOG for any $b \ge 1 + \frac{1}{\poly(n)}$ on random quantum circuits of polynomial size \cite{AG19}. Thus, a physical quantum computer that solves $b$-XHOG for $b \ge 1 + \Omega(1)$ is considered strong evidence of quantum computational supremacy.

In this work, we ask: can an efficient quantum algorithm for $b$-XHOG do substantially better than $b = 2$?\footnote{We thank Scott Aaronson \cite{Aar20} for raising this question, and for suggesting the analogy with the Tsirelson inequality.} That is, what is the largest $b$ for which a polynomial-time quantum algorithm can solve $b$-XHOG on random circuits? Note that the largest $b$ we could hope for is achieved by the optimal sampler that always outputs the string $z$ maximizing $|\braket{z|C|0^n}|^2$. If the random circuits induce a Porter-Thomas distribution on output probabilities, then this solves $b$-XHOG for $b = \Theta(n)$ (see \Cref{fact:max_XHOG} below). However, finding the largest output probability might be computationally difficult even on a quantum computer, which is why we restrict our attention to \textit{efficient} quantum algorithms.

We refer to our problem as the ``quantum supremacy Tsirelson inequality'' in reference to the Bell \cite{Bel64} and Tsirelson \cite{Tsi80} inequalities for quantum nonlocal correlations  (for a modern overview, see \cite{CHTW04}). Under this analogy, the quantity $b$ in XHOG plays a similar role as the probability $p$ of winning some nonlocal game. For example, the Bell inequality for the CHSH game \cite{CHSH69} states that no classical strategy can win the game with probability $p > \frac{3}{4}$; we view this as analogous to the conjectured inability of efficient classical algorithms to solve $b$-XHOG for any $b > 1$. By contrast, a quantum strategy with pre-shared entanglement allows players to win the CHSH game with probability $p = \cos^2\left(\frac{\pi}{8}\right) \approx 0.854 > \frac{3}{4}$. An experiment that wins the CHSH game with probability $p > \frac{3}{4}$, a violation of the Bell inequality, is analogous to an experimental demonstration of $b$-XHOG for $b > 1$ on a quantum computer that establishes quantum computational supremacy. Finally, the Tsirelson inequality for the CHSH game states that any quantum strategy involving arbitrary pre-shared entanglement wins with probability $p \le \cos^2\left(\frac{\pi}{8}\right)$. Hence, an upper bound on $b$ for efficient quantum algorithms is the quantum supremacy counterpart to the Tsirelson inequality. We emphasize that our choice to refer to this as a ``Tsirelson inequality'' is purely by analogy; we do not claim that the question involving quantum supremacy or the techniques one might use to answer it are otherwise related to quantum nonlocal correlations.





\subsection{Our Results}

We study the quantum supremacy Tsirelson inequality in the quantum query (or black box) model. That is, we consider $b$-XHOG with respect to distributions $\mathcal{D}$ that take the following form. We fix a quantum circuit $C$ that queries a classical or quantum oracle $\mathcal{O}$. To sample $C$ from $\mathcal{D}$, we choose $\mathcal{O}$ according to some distribution over oracles. We then ask what is the largest $b$ such that $b$-XHOG with respect to $\mathcal{D}$ is solvable with polynomially many queries to $\mathcal{O}$.

Our motivation for studying this problem in the query model is twofold. First, quantum query results often give useful intuition for what to expect in the real world, and can provide insight into why naive algorithmic approaches fail. Second, we view this as an interesting quantum query complexity problem in its own right. Whereas most other quantum query lower bounds involve decision problems \cite{Amb18} or relation problems \cite{Bel15}, XHOG is more like a weighted, average-case relation problem, because we only require that $|\braket{z|C|0^n}|^2$ be large \textit{on average}. Contrast this with the relation problem considered in \cite{AC17}, where the task is to output a $z$ such that $|\braket{z|C|0^n}|^2$ is greater than some threshold.

Note that there are known quantum query complexity lower bounds for relation problems \cite{AAB+19}, and even relation problems where the output is a quantum state \cite{AMRR11,LR20}. Yet, it is unclear whether existing quantum query lower bound techniques are useful here. Whereas the adversary method tightly characterizes the quantum query complexity of decision problems and state conversion problems \cite{LMRSS11}, it is not known to characterize the query complexity of relation problems, unless they are efficiently verifiable \cite{Bel15}. The adversary method appears to be essentially useless for saying anything about XHOG, which is not efficiently verifiable and is not a relation problem in the traditional sense.\footnote{As we will see later, however, the polynomial method \cite{BBCMdW01} plays an important role in one of our results.}

The XHOG task is well-defined for any distribution of random quantum circuits, so this gives us a choice in selecting the distribution. We focus on three classes of oracle circuits that either resemble random circuits used in practical experiments, or that were previously studied in the context of quantum supremacy. Formal definitions of these oracles (and the associated versions of XHOG) are given in \Cref{sec:oracles}.

\paragraph{Canonical State Preparation Oracles}
Because the linear cross-entropy benchmark for a circuit $C$ depends only on the state $\ket{\psi} := C\ket{0^n}$ produced by the circuit on the all zeros input, it is natural to consider an oracle $\Oracle_\psi$ that prepares a random state $\ket{\psi}$ without leaking additional information about $\ket{\psi}$. Formally, we choose a Haar-random $n$-qubit state $\ket{\psi}$, and fix a canonical state $\ket{\bot}$ orthogonal to all $n$-qubit states.\footnote{\label{foot:1}We can always assume that a convenient $\ket{\bot}$ exists by extending the Hilbert space, if needed. For example, if $\ket{\psi}$ is an $n$-qubit state, a natural choice is to encode $\ket{\psi}$ by $\ket{\psi}\ket{1}$ and to choose $\ket{\bot} = \ket{0^n}\ket{0}$.} Then, we take the oracle $\Oracle_\psi$ that acts as $\Oracle_\psi \ket{\bot} = \ket{\psi}$, $\Oracle_\psi \ket{\psi} = \ket{\bot}$, and $\Oracle_\psi \ket{\varphi} = \ket{\varphi}$ for any state $\ket{\varphi}$ that is orthogonal to both $\ket{\bot}$ and $\ket{\psi}$. Equivalently, $\Oracle_\psi$ is the reflection about the state $\frac{\ket{\psi} - \ket{\bot}}{\sqrt{2}}$. Finally, we let $C$ be the composition of $\Oracle_\psi$ with any unitary that sends $\ket{0^n}$ to $\ket{\bot}$, so that $C\ket{0^n} = \ket{\psi}$. This model is often chosen when proving lower bounds for quantum algorithms that query state preparation oracles \cite{ARU14,AKKT20,BR20}, in part because the ability to simulate $\Oracle_\psi$ follows in a completely black box manner from the ability to prepare $\ket{\psi}$ unitarily without garbage (see \Cref{lem:prep_implies_canonical} below). Hence, the oracle $\Oracle_\psi$ is ``canonical'' in the sense that it is uniquely determined by $\ket{\psi}$ and is not any more powerful than any other oracle that prepares $\ket{\psi}$ without garbage.

\paragraph{Haar-Random Unitaries}
A random polynomial-size quantum circuit $C$ does not behave like a canonical state preparation oracle: $C\ket{x}$ looks like a random quantum state for \textit{any} computational basis state $\ket{x}$, not just $x = 0^n$. Indeed, random quantum circuits are known to information-theoretically approximate the Haar measure in certain regimes \cite{BHH12,HM18}, and it seems plausible that they are also computationally difficult to distinguish from the Haar measure. Thus, one could alternatively model random quantum circuits by Haar-random $n$-qubit unitaries. 

\paragraph{Fourier Sampling Circuits}
Finally, we consider quantum circuits that query a random \textit{classical} oracle. For this, we use $\FS$ circuits, which Aaronson and Chen \cite{AC17} previously studied in the context of proving oracular quantum supremacy for a problem related to XHOG. $\FS$ circuits are defined as $H^{\otimes n}U_f H^{\otimes n}$, where $U_f$ is a phase oracle for a uniformly random Boolean function $f : \{0,1\}^n \to \{-1,1\}$. On the all-zeros input, $\FS$ circuits output a string $z \in \{0,1\}^n$ with probability proportional to the squared Fourier coefficient $\hat{f}(z)^2$. This model has the advantage that in principle, one can prove the corresponding quantum supremacy Bell inequality for classical algorithms given query access to $f$, and that in some cases one can replace $f$ by a pseudorandom function to base quantum supremacy on cryptographic assumptions \cite{AC17}.\\

Our first result is an exponential lower bound on the number of quantum queries needed to solve $(2 + \eps)$-XHOG given either of the two types of quantum oracles that we consider:

\begin{theorem}[Informal version of \Cref{thm:XHOG_O_psi} and \Cref{thm:XHOG_random}]
\label{thm:XHOG_main_informal}
For any $\eps \ge \frac{1}{\poly(n)}$, any quantum query algorithm for $(2 + \eps)$-XHOG with query access to either:
\begin{enumerate}[(1)]
\item a canonical state preparation oracle $\Oracle_\psi$ for a Haar-random $n$-qubit state $\ket{\psi}$, or
\item a Haar-random $n$-qubit unitary,
\end{enumerate} requires at least $\Omega\left(\frac{2^{n/4}}{\poly(n)}\right)$ queries.
\end{theorem}

Recall that, because Haar-random states induce a Porter-Thomas distribution on measurement probabilities, the naive algorithm that outputs a sample from the measurement distribution of the state solves $b$-XHOG for $b \approx 2$. Hence, in the black box setting, \Cref{thm:XHOG_main_informal} implies that it is computationally difficult to substantially beat the naive algorithm for XHOG. We do not know if \Cref{thm:XHOG_main_informal} is quantitatively optimal, but we show in \Cref{thm:XHOG_collision} that a simple algorithm based on the quantum collision finding algorithm \cite{BHT97} solves $(2 + \Omega(1))$-XHOG using $O\left(2^{n/3}\right)$ queries to either oracle. 

Finally, we show that for $\FS$ circuits, the naive algorithm of simply running the circuit is optimal among all $1$-query algorithms:

\begin{theorem}[Informal version of \Cref{thm:fourier_xhog}]
Any $1$-query quantum algorithm for $b$-XHOG with $\FS$ circuits achieves $b \le 3$.
\end{theorem}

Note that the value of $b$ achieved by the naive quantum algorithm for XHOG depends on the distribution of circuits used. In contrast to Haar-random circuits that achieve $b \approx 2$, $\FS$ circuits achieve $b \approx 3$ (see \Cref{prop:fs_xhog}). This stems from the fact that the amplitudes of a Haar-random quantum state are approximately distributed as \textit{complex} normal random variables, whereas the amplitudes of a state produced by a random $\FS$ circuit are approximately distributed as \textit{real} normal random variables.

\subsection{Our Techniques}

The starting point for our proof of the Tsirelson inequality with a canonical state preparation oracle $\Oracle_\psi$ is a result of Ambainis, Rosmanis, and Unruh \cite{ARU14}. It shows that any algorithm that queries $\Oracle_\psi$ can be approximately simulated by a different algorithm that makes no queries, but starts with copies of a resource state that depends on $\ket{\psi}$. This resource state consists of polynomially many (in the number of queries to $\Oracle_\psi$) states of the form $\alpha\ket{\psi} + \beta\ket{\bot}$, i.e. copies of $\ket{\psi}$ in superposition with $\ket{\bot}$. Our strategy is to show that if any algorithm solves $b$-XHOG given this resource state, then a similar algorithm solves $b$-XHOG given copies of $\ket{\psi}$ alone. Then, we prove a lower bound on the number of copies of $\ket{\psi}$ needed to solve $b$-XHOG. To do so, we argue that if $\ket{\psi}$ is Haar-random, then the best algorithm for $b$-XHOG given copies of $\ket{\psi}$ is a simple collision-finding algorithm: measure all copies of $\ket{\psi}$ in the computational basis, and output whichever string $z \in \{0,1\}^n$ appears most frequently in the measurement results. For a Haar-random $n$-qubit state, the chance of seeing \textit{any} collisions is exponentially unlikely, unless the number of copies of $\ket{\psi}$ is exponentially large in $n$, and so this does not do much better than measuring a single copy of $\ket{\psi}$ and outputting the result.

To prove the analogous lower bound for $b$-XHOG with a Haar-random unitary oracle, we show more generally that the canonical state preparation oracles and Haar-random unitary oracles are essentially equivalent as resources, which may be of independent interest. More specifically, we show that for an $n$-qubit state $\ket{\psi}$, using a constant number of queries to $\Oracle_\psi$, one can approximately simulate (to exponential precision) queries to a random oracle that prepares $\ket{\psi}$. By ``random oracle that prepares $\ket{\psi}$,'' we mean an $n$-qubit unitary $U_\psi$ that acts as $U_\psi\ket{0^n} = \ket{\psi}$ but Haar-random everywhere else. We can construct such a $U_\psi$ by taking an arbitrary $n$-qubit unitary that maps $\ket{0^n}$ to $\ket{\psi}$, then composing it with a Haar-random unitary on the $\left(2^n - 1\right)$-dimensional subspace orthogonal to $\ket{0^n}$.

Our lower bound for $\FS$ circuits uses an entirely different technique. We use the polynomial method of Beals et al. \cite{BBCMdW01}, which shows that for any quantum algorithm that makes $T$ queries to a classical oracle, the output probabilities of the algorithm can be expressed as degree-$2T$ polynomials in the variables of the classical oracle. Our key observation is that the average linear XEB score achieved by such a quantum query algorithm can \textit{also} be expressed as a polynomial in the variables of the classical oracle. We further observe that this polynomial is constrained by the requirement that the polynomials representing the output probabilities must be nonnegative and sum to $1$. This allows us to upper bound the largest linear XEB score achievable by the maximum value of a certain linear program, whose variables are the coefficients of the polynomials that represent the output probabilities of the algorithm. To upper bound this quantity, we exhibit a solution to the dual linear program.


\section{Preliminaries}
\label{sec:prelim}

\subsection{Notation}
We use $[N]$ to denote the set $\{1,2,\ldots,N\}$. We use $\id$ to denote the identity matrix (of implicit size). We let $\mathrm{TD}(\rho, \sigma)$ denote the trace distance between density matrices $\rho$ and $\sigma$, and let $\diamondnorm{A}$ denote the diamond norm of a superoperator $A$ acting on density matrices (see \cite{AKN98} for definitions). For a unitary matrix $U$, we use $U \cdot U^\dagger$ to denote the superoperator that maps $\rho$ to $U\rho U^\dagger$. In a slight abuse of notation, if $A$ denotes a quantum algorithm (which may consist of unitary gates, measurements, oracle queries, and initialization of ancilla qubits), then we also use $A$ to denote the superoperator corresponding to the action of $A$ on input density matrices.

\subsection{Oracles for Quantum States}
\label{sec:oracles}
We frequently consider quantum algorithms that query quantum oracles. In this model, a query to a unitary matrix $U$ consists of a single application of either $U$, $U^\dagger$, or controlled versions of $U$ or $U^\dagger$. We also consider quantum algorithms that make queries to \textit{random} oracles. In analogue with the classical random oracle model, such calls are not randomized at each query. Rather, a unitary $U$ is chosen randomly (from some distribution) at the start of the execution of the algorithm, and thereafter all queries for the duration of the algorithm are made to $U$.

We now define several types of unitary oracles that we will use. These definitions (and associated lemmas giving constructions of them) have appeared implicitly or explicitly in prior work, e.g. \cite{ARU14,AKKT20,BR20,ABC+20}. For completeness, we provide proofs of the constructions.

\begin{definition}
For an $n$-qubit quantum state $\ket{\psi}$, the \emph{reflection about $\ket{\psi}$}, denoted $\Refl_\psi$, is the $n$-qubit unitary $\Refl_\psi := \id - 2\ket{\psi}\bra{\psi}$.
\end{definition}

In other words, $\ket{\psi}$ is a $-1$ eigenstate of $\Refl_\psi$, and all states orthogonal to $\ket{\psi}$ are $+1$ eigenstates. Note that some authors define the reflection about $\ket{\psi}$ to be the negation of this operator (e.g. \cite{MNRS07,Rei11,ABC+20}), while others follow our convention (e.g. \cite{BHMT02,KLL+17,AKKT20}). This makes little difference, as these definitions are equivalent up to a global phase (or, if using the controlled versions, equivalent up to a Pauli $Z$ gate).

The following lemma shows that $\Refl_\psi$ can be simulated given any unitary that prepares $\ket{\psi}$ from the all-zeros state, possibly with unentangled garbage.

\begin{lemma}
\label{lem:prep_implies_refl}
Let $U$ be a unitary that acts as $U\ket{0^n}\ket{0^m} = \ket{\psi}\ket{\varphi}$, where $\ket{\psi}$ and $\ket{\varphi}$ are $n$- and $m$-qubit states, respectively. Then one can simulate $T$ queries to the reflection $\Refl_\psi$ using $2T + 1$ queries to $U$.
\end{lemma}

\begin{proof}
Consider the unitary $U(\id - 2\ket{0^{m+n}}\bra{0^{m+n}})U^\dagger$. For any $n$-qubit state $\ket{x}$, the action of this unitary on $\ket{x}\ket{\varphi}$ is equivalent to the action of $\Refl_\psi$ on $\ket{x}\ket{\varphi}$. So, we can simulate $\Refl_\psi$ on $\ket{x}$ as follows: first use one query to $U$ to prepare $\ket{\psi}\ket{\varphi}$ from $\ket{0^{m+n}}$, so that we have a copy of $\ket{\varphi}$. Then, simulate each query to $\Refl_\psi$ using a query to $U$ and $U^\dagger$ to perform $U(\id - 2\ket{0^{m+n}}\bra{0^{m+n}})U^\dagger$ applied to $\ket{x}\ket{\varphi}$, using the copy of $\ket{\varphi}$ prepared in the first step.
\end{proof}

\begin{definition}
For a quantum state $\ket{\psi}$, the \emph{canonical state preparation oracle for $\ket{\psi}$}, denoted $\Oracle_\psi$, is the reflection about the state $\frac{\ket{\psi} - \ket{\bot}}{\sqrt{2}}$, where $\ket{\bot}$ is some canonical state orthogonal to $\ket{\psi}$.
\end{definition}

Unless otherwise specified, we generally assume that if $\ket{\psi}$ is an $n$-qubit state, then $\ket{\bot}$ is orthogonal to the space of $n$-qubit states under a suitable encoding (see \Cref{foot:1}).

The next lemma shows that $\Oracle_\psi$ can be simulated from \textit{any} oracle that prepares $\ket{\psi}$ without garbage:

\begin{lemma}
\label{lem:prep_implies_canonical}
Let $U$ be an $n$-qubit unitary that satisfies $U\ket{0^n} = \ket{\psi}$. Then one can simulate $T$ queries to $\Oracle_\psi$ using $4T + 2$ queries to $U$.
\end{lemma}

\begin{proof}
$\ket{\bot}$ is known, so we may assume that a known unitary $V$ acts as $V\ket{0^n} = \ket{\bot}$. Because $\Oracle_\psi$ is defined as the reflection about $\frac{\ket{\psi} - \ket{\bot}}{\sqrt{2}}$, by \Cref{lem:prep_implies_refl}, it suffices to construct a unitary that prepares any state of the form $\frac{\ket{\psi} - \ket{\bot}}{\sqrt{2}}\ket{\varphi}$ from $\ket{0^n}\ket{0^m}$ using $2$ queries to $U$. The following circuit accomplishes this, with $\ket{\varphi} = \ket{0}$:
$$
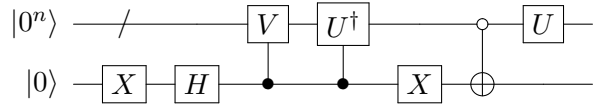
\Qcircuit @C=1em @R=.7em {
\lstick{\ket{0^n}} & /\qw & \qw & \gate{V} & \gate{U^\dagger} & \qw & \ctrlo{1} & \gate{U} & \qw\\
\lstick{\ket{0}} & \gate{X} & \gate{H} & \ctrl{-1} & \ctrl{-1} & \gate{X} & \targ & \qw & \qw
}$$
\end{proof}

We introduce the notion of a \textit{random} state preparation oracle, which, to our knowledge, is new.

\begin{definition}
\label{def:U_psi}
For an $n$-qubit state $\ket{\psi}$ we define a \emph{random state preparation oracle for $\ket{\psi}$}, denoted $U_\psi$, as follows. We fix an arbitrary $n$-qubit unitary $V$ that satisfies $V\ket{0^n} = \ket{\psi}$, then choose a Haar-random unitary $W$ that acts on the $\left(2^n - 1\right)$-dimensional subspace orthogonal to $\ket{0^n}$ in the space of $n$-qubit states. Finally, we set $U_\psi = VW$.
\end{definition}

The invariance of the Haar measure guarantees that this distribution over $U_\psi$ is independent of the choice of $V$, and hence this is well-defined. Note that while we often refer to $U_\psi$ as a single unitary matrix, $U_\psi$ really refers to a \textit{distribution} over unitary matrices. Notice also that if $\ket{\psi}$ is distributed as a Haar-random $n$-qubit state, then $U_\psi$ is distributed as a Haar-random $n$-qubit unitary.

With these definitions in hand, we can now formally define the three versions of the $b$-XHOG task (\Cref{prob:xhog}) that we consider in this paper.

\begin{problem}[$b$-XHOG with canonical state preparation]
\label{prob:xhog_canonical}
Let $\ket{\psi}$ be a Haar-random $n$-qubit state. Given oracle access to $\Oracle_\psi$, output a sample $z \in \{0,1\}^n$ such that:
\[
\E_{\ket{\psi}}\left[|\braket{z|\psi}|^2\right] \ge \frac{b}{2^n}.
\]
\end{problem}

\begin{problem}[$b$-XHOG with a Haar-random oracle]
\label{prob:xhog_haar}
Let $U$ be a Haar-random $n$-qubit unitary. Given oracle access to $U$, output a sample $z \in \{0,1\}^n$ such that:
\[
\E_{U}\left[|\braket{z|U|0^n}|^2\right] \ge \frac{b}{2^n}.
\]
\end{problem}

Equivalently, in the above definition, we can let $\ket{\psi}$ be a Haar-random state, choose $U = U_\psi$, and output a sample $z$ such that:
\[
\E_{\ket{\psi}}\left[|\braket{z|\psi}|^2\right] \ge \frac{b}{2^n}.
\]

\begin{problem}[$b$-XHOG with $\FS$ circuits]
\label{prob:xhog_fourier}
Let $f: \{0,1\}^n \to \{-1,1\}$ be a uniformly random Boolean function of $n$ bits. Let $U_f$ denote the phase oracle for $f$, meaning the unitary transformation which acts as $U_f\ket{x} = f(x) \ket{x}$ for any $x \in \{0,1\}^n$. Given oracle access to $U_f$, output a sample $z \in \{0,1\}^n$ such that:
\[
\E_{f}\left[|\braket{z|H^{\otimes n} U_f H^{\otimes n}|0^n}|^2\right] \ge \frac{b}{2^n},
\]
where $H$ is the Hadamard transformation.
\end{problem}

\subsection{Other Useful Facts}

We use the following formula for the distance between unitary superoperators in the diamond norm.
\begin{fact}[\cite{AKN98}]
\label{fact:diamond_norm}
Let $V$ and $W$ be unitary matrices, and suppose $d$ is the distance between $0$ and the polygon in the complex plane whose vertices are the eigenvalues of $VW^\dagger$. Then
$$\diamondnorm{V \cdot V^\dagger - W \cdot W} = 2\sqrt{1 - d^2}.$$
\end{fact}

Finally, we observe that for a Haar-random $n$-qubit quantum state, the information-theoretically largest linear XEB achievable is $O(n)$, and this is tight.

\begin{fact}\label{fact:max_XHOG}
Let $\ket{\psi}$ be a Haar-random $n$-qubit quantum state. Then:
$$\E_{\ket{\psi}}\left[\max_{z \in \{0,1\}^n}|\braket{z|\psi}|^2\right] = \frac{\Theta(n)}{2^n}.$$
\end{fact}

\begin{proof}[Proof sketch]
For a Haar-random $\ket{\psi}$, the probabilities $|\braket{z|\psi}|^2$ follow a \textit{Porter-Thomas distribution} \cite{AAB+19}, which is to say that they approach i.i.d. exponential random variables with mean $\frac{1}{2^n}$ in the limit. By a well-known result of R\'{e}nyi \cite{Ren53}, the maximum of $N$ i.i.d. exponential random variables with mean $\mu$ is distributed as $\sum_{i=1}^N \frac{E_i}{i}$, where $E_1,\ldots,E_N$ are i.i.d. exponentially distributed with mean $\mu$. In particular, the expected value of the maximum of $N$ i.i.d. exponential random variables with mean $\mu$ is $H_N \cdot \mu$, where $H_N$ is the $N$th harmonic number. So, $\E\left[\max_{z \in \{0,1\}^n}|\braket{z|\psi}|^2\right]$ should approach $\frac{\Theta(n)}{2^n}$, because $H_N \approx \ln N$.

In reality, the probabilities $|\braket{z|\psi}|^2$ are not exactly i.i.d. exponentially distributed, but are distributed according to a Dirichlet distribution (in fact, uniform on the $2^n$-dimensional probability simplex). This distribution can be sampled from as follows: sample $E_1,E_2,\ldots,E_{2^n}$ to be i.i.d. exponential random variables, and set $|\braket{z|\psi}|^2 = \frac{E_z}{\sum_{i=1}^{2^n}E_i}$. The same proof idea still works, essentially because the denominator $\sum_{i=1}^{2^n}E_i$ concentrates well (indeed, the denominator is exponentially close to $1$ with high probability).
\end{proof}

\section{Canonical State Preparation Oracles}
\label{sec:canonical}
In this section, we prove the quantum supremacy Tsirelson inequality for XHOG with a canonical state preparation oracle for a Haar-random state (\Cref{prob:xhog_canonical}). We first sketch the important ideas in the proof. At the heart of our proof is the following lemma, due to Ambainis, Rosmanis, and Unruh \cite{ARU14}. It shows that any quantum algorithm that makes queries to a canonical state preparation oracle $\Oracle_\psi$ can be approximately simulated by a quantum algorithm that makes no queries to $\Oracle_\psi$, and instead receives various copies of $\ket{\psi}$ and superpositions of $\ket{\psi}$ with some canonical orthogonal state.

\begin{lemma}[\cite{ARU14}]
\label{lem:ARU14}
Let $A$ be a quantum query algorithm that makes $T$ queries to $\Oracle_\psi$. Then for any $k$, there is a quantum algorithm $B$ that makes no queries to $\Oracle_\psi$, and a quantum state $\ket{R}$ of the form:
$$\ket{R} := \mathop{\bigotimes}_{j=1}^k \alpha_j \ket{\psi} + \beta_j\ket{\bot}$$
such that for any state $\ket{\varphi}$:
$$\mathrm{TD}(A(\ket{\varphi}\bra{\varphi}), B(\ket{R}\bra{R},\ket{\varphi}\bra{\varphi})) \le O\left(\frac{T}{\sqrt{k}}\right).$$
\end{lemma}
So long as $k \gg T^2$, the output of $B$ will be arbitrarily close to the output of $A$ in trace distance. We will use this and \Cref{fact:max_XHOG} to show that if $A$ solves $b$-XHOG for some $b > 2$, then so does $B$. Then, to prove a lower bound on the number of queries $T$ to $\Oracle_\psi$ needed to solve $b$-XHOG, it suffices to instead lower bound $k$, the number of states of the form $\alpha_j \ket{\psi} + \beta_j\ket{\bot}$ needed to solve $b$-XHOG.

When $\ket{\psi}$ is a Haar-random state, notice that the linear XEB depends only on the \textit{magnitude} of the amplitudes in $\ket{\psi}$; the phases are irrelevant. So, when considering algorithms that attempt to solve $b$-XHOG given only a state $\ket{R}$ of the form used in \Cref{lem:ARU14}, we might as well assume that the algorithm randomly reassigns the phases on $\ket{\psi}$. More formally, define the mixed state $\sigma_R$ as
\begin{equation}
\label{eq:sigma_R}
\sigma_R := \E_{\mathrm{diagonal}\ U}\left[U^{\otimes k}\ket{R}\bra{R} U^{\dagger\otimes k}\right],
\end{equation}
where the expectation is over the diagonal unitaries $U$ such that the entries $\bra{i}U\ket{i}$ are i.i.d. uniformly random complex phases (and by convention, $\braket{\bot|U|\bot} = 1$). Then, the algorithm's average linear XEB score given $\sigma_R$ for a random choice of $\ket{\psi}$ is identical to its average linear XEB score given $\ket{R}$ for a random choice of $\ket{\psi}$, because of the invariance of the Haar measure with respect to phases.

Next, we observe that one can prepare $\sigma_R$ by measuring $k$ copies of $\ket{\psi}$ in the computational basis. We prove this in \Cref{lem:diagonal_symmetrization}. So, when considering algorithms for XHOG that start with $\ket{R}$, it suffices to instead consider algorithms that simply measure $k$ copies of $\ket{\psi}$ in the computational basis. Such algorithms are much easier to analyze, because once we have measured the $k$ copies of $\ket{\psi}$, we can assume (by convexity) that any optimal such algorithm for XHOG outputs a string $z$ deterministically given the $k$ measurement results. And in that case, clearly the optimal strategy is to output whichever $z$ maximizes the posterior expectation of $|\braket{z|\psi}|^2$ given the measurement results. We analyze this strategy in \Cref{lem:XHOG_with_states}, and show that roughly $2^{n/2}$ copies of $\ket{\psi}$ are needed to solve $b$-XHOG for $b$ bounded away from 2. The intuition is that the posterior expectation of $|\braket{z|\psi}|^2$ increases only when we see $z$ at least twice in the measurement results. However, the probability that any two measurement results are the same is tiny---on the order of $2^{-n}$---and so we need to measure at least $2^{n/2}$ copies of $\ket{\psi}$ to see any collisions with decent probability.

We now proceed to proving the necessary lemmas.

\begin{lemma}
\label{lem:diagonal_symmetrization}
Let $\ket{\psi} = \sum_{i=1}^N \psi_i \ket{i}$ be an unknown quantum state, and consider a state $\ket{R}$ of the form:
$$\ket{R} := \mathop{\bigotimes}_{j=1}^k \alpha_j \ket{\psi} + \beta_j\ket{\bot},$$
where $\alpha_j, \beta_j$ are known for $j \in [k]$, and the vectors $\{\ket{1},\ket{2},\ldots,\ket{N},\ket{\bot}\}$ form an orthonormal basis.
Define the mixed state $\sigma_R$ as above \eqref{eq:sigma_R}. Then there exists a protocol to prepare $\sigma_R$ by measuring $k$ copies of $\ket{\psi}$ in the computational basis.
\end{lemma}

To give some intuition, we note that it is simpler to prove \Cref{lem:diagonal_symmetrization} in the case where $\alpha_j = 1$ for all $j$. In that case, $\sigma_R$ can be viewed as an $N^k \times N^k$ density matrix where both the rows and columns are indexed by strings in $[N]^k$. Then, the averaging over diagonal unitaries implies that $\sigma_R$ is obtained from $(\ket{R}\bra{R})^{\otimes k}$ by zeroing out all entries where the index corresponding to the row is not a reordering of the index corresponding to the column. In fact, one can show that $\sigma_R$ is expressible as a mixture of pure states, where each pure state is a uniform superposition over basis states that are reorderings of each other. Moreover, the probability associated with each pure state in this mixture is precisely the probability that one of the reorderings is observed when we measure $k$ copies of $\ket{\psi}$ in the computational basis. So, to prepare $\sigma_R$, it suffices to measure $\ket{\psi}^{\otimes k}$ and then output the uniform superposition over reorderings of the measurement result.

The proof of \Cref{lem:diagonal_symmetrization} is similar, but we instead have to randomly set some of the measurement results to $\bot$ with probability $|\beta_j|^2$.

\begin{proof}[Proof of \Cref{lem:diagonal_symmetrization}]
We first describe the protocol. Define $[N_\bot] := [N] \cup \{\bot\}$. Measure $\ket{\psi}^{\otimes k}$ in the computational basis to obtain a string $\mathbf{x} \in [N]^k$. Then, sample a string $\mathbf{\overline{x}} \in [N_\bot]^k$ by setting
$$\mathbf{\overline{x}}_j = \begin{cases}
\mathbf{x}_j & \text{with probability } |\alpha_j|^2\\
\bot & \text{with probability } |\beta_j|^2
\end{cases}$$
independently for each $j \in [k]$. Let $\mathbf{Z} := \{z \in [N_\bot]^k : z \text{ is a reordering of } \mathbf{\overline{x}}\}$. For each $z \in \mathbf{Z}$ and $j \in [k]$, define
$$\gamma_{zj} := \begin{cases}
\alpha_j & z_j \neq \bot\\
\beta_j & z_j = \bot.
\end{cases}$$
Finally, prepare and output the state
\begin{equation}
\ket{\zeta_{\mathbf{Z}}} := \frac{\sum_{z \in \mathbf{Z}} \left(\prod_{j=1}^k \gamma_{zj} \right)\ket{z}}{\sqrt{\sum_{z \in \mathbf{Z}} \prod_{j=1}^k \gamma_{zj}\gamma_{zj}^*}}.\label{eq:zeta_x}
\end{equation}
This allows us to express the density matrix $\rho_R$ output by this protocol as follows:
\begin{equation}
\rho_R := \sum_{Z \subset [N_\bot]^k} \Pr[\mathbf{Z} = Z] \cdot \ket{\zeta_Z}\bra{\zeta_Z}\label{eq:rho_R}.
\end{equation}
To complete the proof, we want to show that $\rho_R = \sigma_R$. To see that this holds, first consider an entry $\bra{x}\sigma_R\ket{y}$ of $\sigma_R$, where $x, y \in [N_\bot]^k$. It is equal to
\begin{align}
\label{eq:x_sigmaR_y_0}\bra{x}\sigma_R\ket{y} &= \E_{\mathrm{diagonal}\ U} \left[\left(\prod_{j=1}^k \gamma_{xj}\gamma_{yj}^*\right) \cdot \left(\prod_{j: x_j \neq \bot} U_{x_jx_j}\psi_{x_j} \right) \cdot \left(\prod_{j: y_j \neq \bot} U_{y_jy_j}^*\psi_{y_j}^*\right) \right]\\
\label{eq:x_sigmaR_y_1}&= \braket{x|R}\braket{R|y} \cdot \E_{\mathrm{diagonal}\ U} \left[\left(\prod_{j: x_j \neq \bot} U_{x_jx_j} \right) \cdot \left(\prod_{j: y_j \neq \bot} U_{y_jy_j}^*\right) \right]\\
\label{eq:x_sigmaR_y_2}&= \begin{cases}
\braket{x|R}\braket{R|y} & x \text{ is a reordering of } y\\
0 & \text{otherwise.}
\end{cases}
\end{align}
Here, \eqref{eq:x_sigmaR_y_0} and \eqref{eq:x_sigmaR_y_1} are simple calculations that follow from the definitions of $\ket{R}$ and $\sigma_R$. In \eqref{eq:x_sigmaR_y_2}, we use the fact that the entries $U_{ii}$ are independent, uniformly random complex units, and so $\E[U_{ii}^a U_{ii}^{*b}] = \E[U_{ii}^{a-b}]$ is $1$ if $a=b$ and $0$ otherwise, for positive integers $a, b$. Also, if $i \neq j$, then $\E[U_{ii}^a U_{jj}^{*b}] = 0$ unless $a = b = 0$.

Evidently, $\bra{x}\rho_R\ket{y} = \bra{x}\sigma_R\ket{y} = 0$ whenever $x$ is not a reordering of $y$, because $\rho_R$ is a mixture of pure states, each of which is a superposition of basis states that are reorderings of one another. So, it remains to show that $\bra{x}\rho_R\ket{y} = \bra{x}\sigma_R\ket{y} = \braket{x|R}\braket{R|y}$ whenever $x$ is a reordering of $y$. Let $Z := \{z \in [N_\bot]^k : z \text{ is a reordering of } x\}$. Then:
\begin{align}
\label{eq:x_rhoR_y_0}\bra{x}\rho_R\ket{y} &= \Pr[\mathbf{Z} = Z] \cdot \braket{x|\zeta_Z}\braket{\zeta_Z|y}\\
\label{eq:x_rhoR_y_1}&= \left(\sum_{z \in Z} \Pr[\mathbf{\overline{x}} = z] \right) \cdot \braket{x|\zeta_Z}\braket{\zeta_Z|y}\\
\label{eq:x_rhoR_y_2}&= \left(\sum_{z \in Z}\left(\prod_{j=1}^m \gamma_{zj}\gamma_{zj}^*\right)\left(\prod_{j: z_j \neq \bot}\psi_{z_j}\psi_{z_j}^*\right)\right) \cdot \frac{ \prod_{j=1}^m \gamma_{xj}\gamma_{yj}^* }{\sum_{z \in Z} \prod_{j=1}^m \gamma_{zj}\gamma_{zj}^*}\\
\label{eq:x_rhoR_y_3}&= \left(\prod_{j=1}^m \gamma_{xj}\gamma_{yj}^*\right) \cdot \left(\prod_{j: x_j \neq \bot}\psi_{x_j}\psi_{x_j}^* \right)\\
\label{eq:x_rhoR_y_4}&= \braket{x|R}\braket{R|y}\\
\label{eq:x_rhoR_y_5}&= \bra{x}\sigma_R\ket{y}.
\end{align}
Here, \eqref{eq:x_rhoR_y_0} holds because $\ket{\zeta_Z}$ and $\ket{\zeta_{Z'}}$ have disjoint support when $Z \cap Z' = \emptyset$; \eqref{eq:x_rhoR_y_1} holds by definition of $\mathbf{Z}$; \eqref{eq:x_rhoR_y_2} holds by definitions of $\mathbf{\overline{x}}$ and $\ket{\zeta_Z}$; \eqref{eq:x_rhoR_y_3} is a simplification; \eqref{eq:x_rhoR_y_4} holds by definition of $\ket{R}$, and \eqref{eq:x_rhoR_y_5} follows from \eqref{eq:x_sigmaR_y_2}, because $x$ was assumed to be a reordering of $y$. 
\end{proof}

Combining \Cref{lem:ARU14} and \Cref{lem:diagonal_symmetrization}, we have reduced the problem of lower bounding the number of $\Oracle_\psi$ queries needed to solve $b$-XHOG, to lower bounding the number of copies of $\ket{\psi}$ needed to solve $b$-XHOG. The next lemma lower bounds this latter quantity.

\begin{lemma}
\label{lem:XHOG_with_states}
Let $\ket{\psi}$ be a Haar-random $n$-qubit quantum state. Consider a quantum algorithm that receives as input $\ket{\psi}^{\otimes k}$ and outputs a string $z \in \{0,1\}^n$. Then:
$$\E_{\ket{\psi},z}\left[|\braket{z|\psi}|^2\right] \le \frac{2}{2^n} + \frac{O(k^2)}{4^n}.$$
\end{lemma}

\begin{proof}
Let $\ket{R} = \ket{\psi}^{\otimes k}$. As we have argued above, the algorithm achieves the same linear XEB score on average if it instead begins with the mixed state $\sigma_R$ defined in \eqref{eq:sigma_R}, because of the invariance of the Haar measure with respect to phases. By \Cref{lem:diagonal_symmetrization}, the algorithm can prepare $\sigma_R$ by measuring $\ket{R}$ in the computational basis. By a convexity argument, we can assume that the algorithm outputs $z$ deterministically given the measurement results.

Suppose the measurement results are $z_1,z_2,\ldots,z_k$. Clearly, the choice of $z$ that maximizes $\E\left[|\braket{z|\psi}|^2\right]$ is whichever $z$ appears most frequently in $z_1,z_2,\ldots,z_k$ (with ties broken arbitrarily): the probabilities $|\braket{i|\psi}|^2$ are distributed according to a $\mathrm{Dir}(1,1,\ldots,1)$ distribution, so we can easily compute the posterior expectations $\E\left[|\braket{i|\psi}|^2 \mid z_1,z_2,\ldots,z_k\right]$. So, it suffices to bound $\E\left[|\braket{z|\psi}|^2\right]$ for the algorithm that chooses $z$ to be the most frequent measurement result.

Let $m$ be a random variable that denotes the frequency of the chosen $z$. Then
\begin{align}
\label{eq:0}
\E\left[|\braket{z|\psi}|^2\right] &= \E\left[\E\left[|\braket{z|\psi}|^2 \mid m \right]\right]\\ \label{eq:1}
&= \E\left[\frac{1 + m}{2^n + k}\right]\\ \label{eq:2}
&\le \frac{1}{2^n} + \E\left[\frac{m}{2^n}\right]\\ \label{eq:3}
&\le \frac{1}{2^n} + \E\left[\frac{1 + \sum_{i \neq j} \id[z_i = z_j]}{2^n}\right]\\ \label{eq:4}
&= \frac{2}{2^n} + \sum_{i \neq j}\frac{\Pr[z_i = z_j]}{2^n}\\ \label{eq:5}
&= \frac{2}{2^n} + \binom{k}{2}\frac{2}{2^n(2^n + 1)}\\
&\le \frac{2}{2^n} + \frac{O(k^2)}{4^n}.
\end{align}
Here, \eqref{eq:0} is valid by the law of total expectation. \eqref{eq:1} substitutes the formula for the posterior expectation of a Dirichlet distribution. \eqref{eq:2} is valid by linearity of expectation. In \eqref{eq:3}, we use the crude upper bound that $m$ is at most one more than the number of pairwise collisions in $z_1,\ldots,z_k$ (which is tight when the number of collisions is 0 or 1). \eqref{eq:4} is valid by linearity of expectation. \eqref{eq:5} expands the sum, and computes the collision probabilities in terms of moments of the underlying $\mathrm{Dir}(1,1,\ldots,1)$ prior distribution.
\end{proof}

We note that one should not expect \Cref{lem:XHOG_with_states} to be tight for large $k$ (say, $k = \Omega\left(2^{n/2}\right)$). For example, to achieve $b = 4$, we need at least enough samples to see $m \ge 3$ with good probability. But $\Pr[m \ge 3]$ is negligible unless $k = \Omega\left(2^{2n/3}\right)$. More generally, a tight bound on the number of copies of $\ket{\psi}$ needed to achieve a particular value of $b$ seems closely related to the number of measurements of $\ket{\psi}$ needed to see $m \ge b - 1$. This is like a sort of ``balls into bins'' problem \cite{JK77,RS98} with $k$ balls and $2^n$ bins in which we want to bound the probability that the maximum load of any bin exceeds $m$, but where the probabilities associated to each bin follow a Dirichlet prior rather than being uniform.

We finally have the tools to prove the main result of this section.

\begin{theorem}
\label{thm:XHOG_O_psi}
Any quantum query algorithm for $(2 + \eps)$-XHOG with query access to $\Oracle_\psi$ for a Haar-random $n$-qubit state $\ket{\psi}$ requires $\Omega\left(\frac{2^{n/4}\eps^{5/4}}{n}\right)$ queries.
\end{theorem}

\begin{proof}
Consider a quantum algorithm $A$ that makes $T$ queries to $\Oracle_\psi$ and solves $(2 + \eps)$-XHOG. Choose $k = \frac{c^2T^2n^2}{\eps^2}$ in \Cref{lem:ARU14} for a constant $c$ to be chosen later. By \Cref{lem:ARU14}, there is a quantum algorithm $B$ that makes no queries to $\Oracle_\psi$ and instead starts with a state $\ket{R}$ (depending on $\ket{\psi}$) such that the trace distance between the output of $A$ and $B$ is at most $O\left(\frac{\eps}{cn}\right)$ for every $\ket{\psi}$. In particular, if we view $\ket{\psi}$ as fixed, then the total variation distance between the outputs $z_A$ and $z_B$ of $A$ and $B$, respectively, (as probability distributions over $\{0,1\}^n$) is at most $O\left(\frac{\eps}{cn}\right)$. Hence, for every $\ket{\psi}$, we may write:
\begin{align*}
\E_{z_A}\left[ |\braket{z_A|\psi}|^2 \right] - \E_{z_B}\left[ |\braket{z_B|\psi}|^2 \right] &= \sum_{z \in \{0,1\}^n} |\braket{z|\psi}|^2 \cdot \left( \Pr[z_A = z] - \Pr[z_B = z] \right)\\
&\le \sum_{z\in \{0,1\}^n} |\braket{z|\psi}|^2 \cdot \left| \Pr[z_A = z] - \Pr[z_B = z] \right| \\
&\le \max_{z \in \{0,1\}^n} |\braket{z|\psi}|^2 \cdot \sum_{z' \in \{0,1\}^n} \left| \Pr[z_A = z'] - \Pr[z_B = z'] \right|\\
&\le \max_{z \in \{0,1\}^n}  |\braket{z|\psi}|^2 \cdot O\left(\frac{\eps}{cn}\right),
\end{align*}
because the sum in the penultimate inequality is twice the total variation distance between $z_A$ and $z_B$. \Cref{fact:max_XHOG} states that for a Haar-random $\ket{\psi}$, $\E_{\ket{\psi}}\left[\max_{z \in \{0,1\}^n}|\braket{z|\psi}|^2\right] \le \frac{O(n)}{2^n}$. So, for a Haar-random $\ket{\psi}$, we have
$$\E_{\ket{\psi},z_A}\left[ |\braket{z_A|\psi}|^2 \right] - \E_{\ket{\psi},z_B}\left[ |\braket{z_B|\psi}|^2 \right] \le O\left(\frac{\eps}{c2^n}\right).$$
In particular, if we choose $c$ sufficiently large, then $B$ solves $\left(2 + \frac{\eps}{2}\right)$-XHOG.

Because of the invariance of the Haar measure with respect to phases, $B$ still solves $\left(2 + \frac{\eps}{2}\right)$-XHOG if the pure state $\ket{R}$ is replaced with the mixed state $\sigma_R$ defined in \eqref{eq:sigma_R}. By \Cref{lem:diagonal_symmetrization}, this implies the existence of an algorithm that solves $\left(2 + \frac{\eps}{2}\right)$-XHOG given $k$ copies of $\ket{\psi}$. By \Cref{lem:XHOG_with_states}, such an algorithm must satisfy:
$$\frac{\eps}{2} \le \frac{O(k^2)}{2^n}.$$
Plugging in $k$ gives the desired lower bound on $T$:
$$\frac{\eps}{2} \le O\left(\frac{T^4n^4}{2^n\eps^4}\right)$$
\[T \ge \Omega\left(\frac{2^{n/4}\eps^{5/4}}{n}\right).\qedhere\]
\end{proof}

Lastly, we give an upper bound on the number of queries needed to nontrivially beat the naive algorithm for XHOG with $\Oracle_\psi$. In fact, the following algorithm works with \textit{any} oracle that prepares a Haar-random state (including a Haar-random unitary), because the algorithm only needs copies of $\ket{\psi}$ and the ability to perform the reflection $\Refl_\psi$. We thank Scott Aaronson for suggesting this approach based on quantum collision-finding.

\begin{theorem}
\label{thm:XHOG_collision}
There is a quantum algorithm for $(2 + \Omega(1))$-XHOG that makes $O\left(2^{n/3}\right)$ queries to a state preparation oracle for a Haar-random $n$-qubit state $\ket{\psi}$.
\end{theorem}

\begin{proof}
The quantum algorithm is essentially equivalent to the collision-finding algorithm of Brassard, H\o{}yer, and Tapp \cite{BHT97}. We proceed by measuring $k = 2^{n/3}$ copies of $\ket{\psi}$ in the computational basis, with results $z_1, z_2, \ldots, z_k \in \{0,1\}^n$. If any string appears twice in $z_1, z_2, \ldots, z_k$, we output the first such collision. Otherwise, we perform quantum amplitude amplification \cite{BHMT02} on another copy of $\ket{\psi}$, where the ``good'' subspace is spanned by $z_1, z_2, \ldots, z_k$. This uses the reflection $\Refl_\psi$, which can be simulated using a constant number of queries to any oracle that prepares $\ket{\psi}$ (see \Cref{lem:prep_implies_refl}). Finally, we measure and output the result of the amplitude amplification; call this result $z_{k+1}$. For the purpose of analyzing this algorithm, we say that the algorithm ``succeeds'' if it either finds a collision in $z_1, z_2, \ldots, z_k$, or if $z_{k+1}$ is contained in the good subspace.

We first argue that for any $\ket{\psi}$, $O\left(2^{n/3}\right)$ queries to $\Refl_\psi$ (for amplitude amplification) are sufficient for the algorithm to succeed with high probability. For this part of the analysis, we view $\ket{\psi}$ as fixed, and consider only the randomness of the algorithm. Notice that for each $1 \le i \le k$, $\Pr\left[|\braket{z_i|\psi}|^2 \ge \frac{1}{2^{n+1}}\right] \ge \frac{1}{2}$, because at most half of the probability mass of the output distribution of $\ket{\psi}$ can be placed on inputs for which the output probability is less than $\frac{1}{2^{n+1}}$, because there are only $2^n$ possible outputs. Thus, by a Chernoff bound, we have that:
$$\Pr\left[\sum_{i=1}^k |\braket{z_i|\psi}|^2 \ge \frac{k}{2^{n + 2}}\right] \ge 1 - \exp(O(k)).$$
In particular, with probability $1 - \exp(O(k))$, either the algorithm finds a collision in $z_1,z_2,\ldots,z_k$, or else $O\left(\sqrt{\frac{2^n}{k}}\right) = O\left(2^{n/3}\right)$ applications of $\Refl_\psi$ within the amplitude amplification subroutine are sufficient to measure a good string with arbitrarily high constant probability. So overall, the algorithm can be assumed to succeed with arbitrarily high constant probability.

Next, we argue that the algorithm outputs a string $z$ such that $\E\left[|\braket{z|\psi}|^2\right] \ge \frac{2 + \Omega(1)}{2^n}$. Suppose that instead of performing amplitude amplification at the end, we just measured one additional copy of $\ket{\psi}$ (still calling the result $z_{k+1}$) and output $z_{k+1}$ if there were no collisions in $z_1,z_2,\ldots,z_k$. Then notice that, conditional on this modified algorithm's success, the expected XEB score is at least $\frac{3 - o(1)}{2^n}$. In symbols, we claim that:
$$\E\left[|\braket{z|\psi}|^2 \mid \mathrm{success}\right] \ge \frac{3}{2^n + k + 1}$$
for this modified algorithm, because conditional on success, $z$ was observed at least twice in $z_1,z_2,\ldots,z_{k+1}$, so $\frac{3}{2^n + k + 1}$ is a lower bound on the posterior expectation of the underlying Dirichlet prior distribution on the output probabilities of $\ket{\psi}$. But now, we claim that $\E\left[|\braket{z|\psi}|^2 \mid \mathrm{success}\right]$ is the same for both the modified algorithm and the original algorithm that uses amplitude amplification. The reason is that amplitude amplification preserves conditional probabilities: the conditional probability distribution of $z_{k+1}$ is exactly the same in both algorithms, when conditioned on measuring in the good subspace. So overall, we have that:
\begin{align*}
\E\left[|\braket{z|\psi}|^2\right] &= \E\left[|\braket{z|\psi}|^2 \mid \mathrm{success}\right] \cdot \Pr[\mathrm{success}] + \E\left[|\braket{z|\psi}|^2 \mid \mathrm{failure}\right] \cdot \Pr[\mathrm{failure}]\\
&\ge \E\left[|\braket{z|\psi}|^2 \mid \mathrm{success}\right] \cdot \Pr[\mathrm{success}]\\
&= \frac{3 - o(1)}{2^n} \cdot (1 - p)\\
&\ge \frac{2 + \Omega(1)}{2^n},
\end{align*}
where $p$ is the arbitrarily small constant failure probability of amplitude amplification.
\end{proof}

We remark that a sharper analysis could most likely improve the above algorithm from solving $(2 + \Omega(1))$-XHOG to solving $(3 - o(1))$-XHOG, while still using the same number of queries. For most Haar-random states $\ket{\psi}$, the probability of measuring in the ``good'' subspace should concentrate very well. As a result, it should be possible to fix some $T(n)$ such that running exactly $T(n)$ iterations of Grover's algorithm ensures finding a marked with high probability, rather than constant probability.

\section{Random State Preparation Oracles}
\label{sec:random}
In this section, we show that a canonical state preparation oracle and a random state preparation oracle are essentially equivalent, and use it to prove the quantum supremacy Tsirelson inequality for XHOG with a Haar-random oracle (\Cref{prob:xhog_haar}).

By \Cref{lem:prep_implies_canonical}, for a state $\ket{\psi}$, query access to a random state preparation oracle $U_\psi$ implies query access to the canonical state preparation oracle $\Oracle_\psi$ with constant overhead. The reverse direction is less obvious. We know from the definition of $U_\psi$ (\Cref{def:U_psi}) that one can simulate $U_\psi$ given \textit{any} $n$-qubit unitary $V$ that prepares $\ket{\psi}$ from $\ket{0^n}$. So, it is tempting to let $V = \Oracle_\psi$ with $\ket{\bot} = \ket{0^n}$ to argue that $\Oracle_\psi$ allows simulating $U_\psi$. However, this is only possible if $\ket{0^n}$ is orthogonal to $\ket{\psi}$. And while we previously argued that we can always find a canonical state $\ket{\bot}$ that is orthogonal to $\ket{\psi}$ (\Cref{foot:1}), this requires extending the Hilbert space, so that $\Oracle_\psi$ no longer acts on $n$ qubits!

To address this, imagine that we knew an explicit $n$-qubit state $\ket{\psi^\bot}$ orthogonal to $\ket{\psi}$. Notice that we could perfectly swap $\ket{\psi}$ and $\ket{\psi^\bot}$: the composition $\Oracle_\psi\Oracle_{\psi^\bot}\Oracle_\psi$ sends $\ket{\psi}$ to $\ket{\psi^\bot}$, $\ket{\psi^\bot}$ to $\ket{\psi}$, and acts trivially on all states orthogonal to $\ket{\psi}$ and $\ket{\psi^\bot}$. In particular, this swaps $\ket{\psi}$ and $\ket{\psi^\bot}$ while acting only on the space of $n$-qubit states. Next, if we know $\ket{\psi^\bot}$ explicitly, we can certainly come up with an $n$-qubit unitary that sends $\ket{0^n}$ to $\ket{\psi^\perp}$. By composing such a unitary with $\Oracle_\psi\Oracle_{\psi^\bot}\Oracle_\psi$, we are left with an $n$-qubit unitary that sends $\ket{0^n}$ to $\ket{\psi}$. This is sufficient to construct $U_\psi$, by \Cref{def:U_psi}.

While we do not necessarily have such a state $\ket{\psi^\bot}$, a \textit{random} $n$-qubit state $\ket{\varphi}$ will be exponentially close to such a $\ket{\psi^\bot}$ with overwhelming probability. The next theorem shows that we can use this observation to \textit{approximately} simulate $U_\psi$ given $\Oracle_\psi$, by going through the steps above and keeping track of deviation from the ideal construction in terms of $\braket{\psi|\varphi}$.

\begin{theorem}\label{thm:standard->random}
Let $\ket{\psi}$ be an $n$-qubit state. Consider a quantum query algorithm $A$ that makes $T$ queries to $U_\psi$. Then there is a quantum query algorithm $B$ that makes $2T$ queries to $\Oracle_\psi$ such that:
$$\diamondnorm{\E_{U_\psi}\left[A\right] - B} \le \frac{10T + 4}{2^{n/2}}.$$
\end{theorem}

\begin{proof}
Without loss of generality, assume $\ket{\bot}$ is orthogonal to all $n$-qubit states. Let $\ket{\varphi}$ be a Haar-random $n$-qubit state, and let $V$ be an arbitrary $n$-qubit unitary that satisfies $V\ket{0^n} = \ket{\varphi}$. Write $\ket{\varphi} = \alpha\ket{\psi^\bot} + \beta\ket{\psi}$, where $\ket{\psi^\bot}$ is some $n$-qubit state orthogonal to $\ket{\psi}$, with the phase chosen so that $\alpha$ is real and nonnegative. Note that $\beta = \braket{\psi|\varphi}$.

Suppose we had an oracle $V'$ acting on $n$ qubits such that $V'\ket{0^n} = \ket{\psi^\bot}$. Then we could appeal to \Cref{lem:prep_implies_canonical} to simulate an oracle $\Oracle_{\psi^\bot}$ that reflects about the state $\frac{\ket{\psi^\bot} - \ket{\bot}}{\sqrt{2}}$ using queries to $V'$. Then the composition $\Oracle_\psi\Oracle_{\psi^\bot}\Oracle_\psi$ would swap $\ket{\psi}$ and $\ket{\psi^\bot}$, while acting only on the space of $n$-qubit states. Furthermore, we would have that $\Oracle_\psi\Oracle_{\psi^\bot}\Oracle_\psi V'\ket{0^n} = \ket{\psi}$, where $\Oracle_\psi\Oracle_{\psi^\bot}\Oracle_\psi V'$ acts on $n$ qubits. So, by \Cref{def:U_psi}, we could simulate $U_\psi$ perfectly by choosing a random $(2^n-1)$-dimensional unitary $W$ and replacing calls to $U_\psi$ with $\Oracle_\psi\Oracle_{\psi^\bot}\Oracle_\psi V'W$.

Unfortunately, we do not have such an oracle $V'$; we only have $V$. However, we can show that there exists an oracle $V'$ that is \textit{close} to $V$, so if we replace all occurrences of $V'$ with $V$, the resulting unitary we get is close to a random state preparation oracle for $\ket{\psi}$. Specifically, we take $R$ to be a rotation in the 2-dimensional space spanned by $\ket{\psi}$ and $\ket{\psi^\bot}$ that satisfies $R\ket{\varphi} = \ket{\psi^\bot}$. Then, we let $V' = RV$.

$R$ is a rotation by angle $\theta = \arccos(\alpha)$ in this 2-dimensional subspace, and acts as the identity elsewhere. So, $R$ has eigenvalues $e^{i \theta}$, $e^{-i\theta}$, and $1$. The assumption that $\alpha \ge 0$ implies $\theta \le \frac{\pi}{2}$, so by \Cref{fact:diamond_norm},
\begin{equation}
\label{eq:vv'}
\diamondnorm{V \cdot V^\dagger - V' \cdot V'^\dagger} = 2\sqrt{1 - \cos^2(\theta)} = 2\sin \theta = 2|\braket{\psi|\varphi}|.
\end{equation}

\Cref{lem:prep_implies_canonical} shows that $V'$ (or more precisely, controlled-$V'$ or its inverse) is used $4T + 2$ times in implementing $T$ queries to $\Oracle_{\psi^\bot}$, which means we need $5T + 2$ applications of $V'$ to implement $T$ queries to $\Oracle_\psi\Oracle_{\psi^\bot}\Oracle_\psi V'$.

Let $B_{\psi^\bot}$ denote the quantum algorithm that simulates $A$ using $\Oracle_\psi\Oracle_{\psi^\bot}\Oracle_\psi V'W$ (for a random choice of $W$) in place of $U_\psi$, and let $B_{\varphi}$ denote the quantum algorithm that simulates $B_{\psi^\bot}$ using $V$ in place of $V'$. Then
\begin{align}
\label{eq:d0}\diamondnorm{\E_{U_\psi}[A] - B_{\varphi}} &= \diamondnorm{B_{\psi^\bot} - B_{\varphi}}\\
\label{eq:d1}&\le (5T + 2)\diamondnorm{V \cdot V^\dagger - V' \cdot V'^\dagger}\\
\label{eq:d2}&= (10T + 4)|\braket{\psi|\varphi}|,
\end{align}
where \eqref{eq:d0} holds because $\E_{U_\psi}[A]$ and $B_{\psi^\bot}$ are equivalent as superoperators; \eqref{eq:d1} holds by the subadditivity of the diamond norm under composition, because $B_{\psi^\bot}$ queries $V'$ a total of $5T + 2$ times; and \eqref{eq:d2} substitutes \eqref{eq:vv'}.

Finally, let $B = \E_{\ket{\varphi}}\left[B_\varphi\right]$ (i.e. run $B_\varphi$ for a Haar-random choice of $\ket{\varphi}$). Then
\begin{align}
\label{eq:lem_haar_final1}
\diamondnorm{\E_{U_\psi}[A] - B} &= \diamondnorm{\E_{U_\psi}[A] - \E_{\ket{\varphi}}\left[B_\varphi\right]}\\
\label{eq:lem_haar_final2} &\le \E_{\ket{\varphi}}\left[\diamondnorm{\E_{U_\psi}[A] - B_\varphi} \right]\\
\label{eq:lem_haar_final3} &\le \E_{\ket{\varphi}}\left[(10T + 4)|\braket{\psi|\varphi}|\right]\\
\label{eq:lem_haar_final4} &= (10T + 4)\E_{\ket{\varphi}}\left[\frac{\sum_{i=1}^{2^n}|\braket{i|\varphi}|}{2^n}\right]\\
\label{eq:lem_haar_final5} &\le \frac{10T + 4}{2^n}\max_{\ket{\varphi}}\left[\sum_{i=1}^{2^n}|\braket{i|\varphi}|\right]\\
\label{eq:lem_haar_final6} &\le \frac{10T + 4}{2^{n/2}},
\end{align}
where \eqref{eq:lem_haar_final1} holds by the definition of $B$; \eqref{eq:lem_haar_final2} holds by Jensen's inequality because the diamond norm, like every norm, is convex; \eqref{eq:lem_haar_final3} substitutes \eqref{eq:d2}; \eqref{eq:lem_haar_final4} holds by symmetry (the choice of orthonormal basis $\{\ket{i} : i \in [2^n]\}$ is arbitrary); \eqref{eq:lem_haar_final5} trivially upper bounds \eqref{eq:lem_haar_final4}; and \eqref{eq:lem_haar_final6} holds because the $1$-norm of an $n$-qubit quantum state is at most $2^{n/2}$ (maximized by a uniform superposition).
\end{proof}

The above theorem implies that the oracle $\Oracle_\psi$ in \Cref{thm:XHOG_O_psi} can be replaced by a Haar-random $n$-qubit unitary.

\begin{theorem}
\label{thm:XHOG_random}
Any quantum query algorithm for $(2 + \eps)$-XHOG with query access to $U_\psi$ for a Haar-random $n$-qubit state $\ket{\psi}$ (i.e. a Haar-random $n$-qubit unitary) requires $\Omega\left(\frac{2^{n/4}\eps^{5/4}}{n}\right)$ queries.
\end{theorem}

\begin{proof}
Consider a quantum algorithm $A$ that makes $T$ queries to $U_\psi$ and solves $(2 + \eps)$-XHOG. Let $c$ be a constant to be chosen later. If $T > c\frac{2^{n/2}\eps}{n}$, then we are done, because we can always assume that $\eps \le O(n)$ (\Cref{fact:max_XHOG}), so $\frac{2^{n/2}\eps}{n} \ge \frac{2^{n/4}\eps^{5/4}}{n}$ for sufficiently large $n$. In the complementary case, suppose that $T \le c\frac{2^{n/2}\eps}{n}$. By \Cref{thm:standard->random} and the definition of the diamond norm, there is a quantum query algorithm $B$ that makes $2T$ queries to $\Oracle_\psi$ such that the trace distance between the output of $A$ (averaged over the choice of $U_\psi$) and $B$ is at most $\frac{10T + 4}{2^{n/2}} \le \frac{14T}{2^{n/2}} \le \frac{14c\eps}{n}$ for every $\ket{\psi}$. By an argument involving \Cref{fact:max_XHOG} similar to the one used in the proof of \Cref{thm:XHOG_O_psi}, we conclude that if $c$ is a sufficiently small constant, then $B$ solves $\left(2 + \frac{\eps}{2}\right)$-XHOG (with a canonical state preparation oracle for a Haar-random state). By \Cref{thm:XHOG_O_psi}, this implies $T = \Omega\left(\frac{2^{n/4}\eps^{5/4}}{n}\right)$.
\end{proof}

\section{Fourier Sampling Circuits}
\label{sec:fourier}
In this section, we prove the quantum supremacy Tsirelson inequality for single-query algorithms over $\FS$ circuits (\Cref{prob:xhog_fourier}).

Throughout this section, we let $N = 2^n$, and let $\Fn := \left\{f : \{0,1\}^n \to \{-1, 1\}\right\}$ denote the set of all $n$-bit Boolean functions. Given a function $f \in \Fn$, we define the Fourier coefficient
$$\hat{f}(z) := \frac{1}{2^n} \sum_{x \in \{0,1\}^n} f(x) (-1)^{x \cdot z}$$
for each $z \in \{0,1\}^n$. We also define the characters $\chi_z : \{0,1\}^n \to \{-1,1\}$ for each $z \in \{0,1\}^n$:
$$\chi_z(x) := (-1)^{x \cdot z}.$$

Given oracle access to a function $f \in \Fn$, the $\FS$ quantum circuit for $f$ consists of a layer of Hadamard gates, then a single query to $f$, then another layer of Hadamard gates, so that the resulting circuit samples a string $z \in \{0,1\}^n$ with probability $\hat{f}(z)^2$. In the context of XHOG, we consider the distribution of $\FS$ circuits where the oracle $f$ is chosen uniformly at random from $\Fn$.

\begin{proposition}
\label{prop:fs_xhog}
$\FS$ circuits over $n$ qubits solve $(3 - \frac{2}{2^n})$-XHOG.
\end{proposition}

\begin{proof}
Because the circuit samples $z$ with probability $\hat{f}(z)^2$, the expected linear XEB score is:
\begin{align}
\label{eq:fs_xhog_1}\E_{f \in \Fn}\left[ \sum_{z \in \{0,1\}^n} \hat{f}(z)^4 \right] &= \E_{f \in \Fn}\left[ \sum_{z \in \{0,1\}^n} \widehat{f \cdot \chi_z}(0^n)^4 \right]\\
\label{eq:fs_xhog_2}&= 2^n \E_{f \in \Fn}\left[ \hat{f}(0^n)^4 \right]\\
\label{eq:fs_xhog_3}&= 2^n \left(\frac{2}{2^n}\right)^4 \E\left[\left(B(2^n, 1/2) - \E\left[B(2^n, 1/2)\right]\right)^4 \right]\\
\label{eq:fs_xhog_4}&= \frac{3 - \frac{2}{2^n}}{2^n},
\end{align}
where \eqref{eq:fs_xhog_1} applies the substitution $\hat{f}(z) = \widehat{f \cdot \chi_z}(0^n)$; \eqref{eq:fs_xhog_2} is valid because if $f$ is uniform over $\Fn$ then so is $f \cdot \chi_z$; \eqref{eq:fs_xhog_3} uses the fact that $\frac{2^n}{2}(\hat{f}(0^n) + 1)$ is binomially distributed with $2^n$ trials and success probability $\frac{1}{2}$; and \eqref{eq:fs_xhog_4} uses the formula $Np(1-p)(1 + (3N - 6)p(1 - p))$ for the $4$th central moment of a binomial distribution with $N$ trials and success probability $p$.
\end{proof}

The remainder of this section constitutes the proof of the following theorem, which shows the optimality of the 1-query algoritm for XHOG with $\FS$ circuits:

\begin{theorem}\label{thm:fourier_xhog}
Any $1$-query algorithm for $b$-XHOG over $n$-qubit $\FS$ circuits satisfies $b \le 3 - \frac{2}{2^n}$.
\end{theorem}

To prove \Cref{thm:fourier_xhog}, we use the polynomial method of Beals et al. \cite{BBCMdW01}. Consider a quantum query algorithm that makes $T$ queries to $f \in \Fn$ and outputs a string $z \in \{0,1\}^n$. The polynomial method implies that for each $z \in \{0,1\}^n$, the probability that the algorithm outputs $z$ can be expressed as a real multilinear polynomial of degree $2T$ in the bits of $f$. We write such a polynomial as:
$$p_z(f) = \sum_{S \subset \{0,1\}^n, |S| \le 2T} c_{z,S} \cdot \prod_{x \in S} f(x).$$
Then, the expected XEB score of this quantum query algorithm is given by:
\begin{equation}
\label{eq:fs_xeb}
\frac{1}{2^N} \sum_{f \in \Fn} \sum_{z \in \{0,1\}^n} p_z(f) \cdot \hat{f}(z)^2.
\end{equation}
Our key observation is that the quantity \eqref{eq:fs_xeb} is linear in the coefficients $c_{z, S}$. This allows us to express the largest XEB score achievable by polynomials of degree $2T$ as a linear program, with the constraints that the polynomials $\{p_z(f) : z \in \{0,1\}^n\}$ must represent a probability distribution. Then, the objective value of the linear program can be upper bounded by giving a solution to the dual linear program. We can write the linear program as follows:
\begin{equation}
\label{eq:primal_lp_unsimplified}
\boxed{\begin{array}{lll} 
    \text{max}  & \displaystyle\frac{1}{2^N} \sum_{f \in \Fn} \sum_{z \in \{0,1\}^n} p_z(f) \cdot \hat{f}(z)^2    \\
    \text{subject to} & p_z(f) \ge 0 & \text{ for each } f \in \Fn; z \in \{0, 1\}^n\\
    &\displaystyle\sum_{z \in \{0,1\}^n} p_z(f) = 1 & \text{ for each } f \in \Fn\\
    &c_{z,S} \in \mathbb{R} & \text{ for each } z \in \{0, 1\}^n; 0 \le |S| \le 2T
\end{array}}
\end{equation}
Before giving a solution to (or even writing down) the dual linear program, we will first show that the primal linear program can be simplified considerably.

We first argue that one can apply a sort of symmetrization to reduce the number of variables. Consider a solution to the linear program \eqref{eq:primal_lp_unsimplified} in terms of polynomials $p_z$, and define:
$$p'_z(f) = \frac{1}{N} \sum_{y \in \{0,1\}^n} p_{y \oplus z}(f \cdot \chi_y).$$
Then we claim that the polynomials $p'_z$ are also a solution to the linear program with the same objective value. The intuition is that $\hat{f}(z) = \widehat{f \cdot \chi_y}(y \oplus z)$, so we might as well assume that the probability of outputting $z$ on $f$ is the same as the probability of outputting $y \oplus z$ on $f \cdot \chi_y$, by averaging over the possible choices of $y$. We verify that the objective value is:
\begin{align*}
\frac{1}{2^N} \sum_{f \in \Fn} \sum_{z \in \{0,1\}^n} p'_z(f) \cdot \hat{f}(z)^2 &= \frac{1}{N2^N} \sum_{f \in \Fn} \sum_{z \in \{0,1\}^n} \sum_{y \in \{0,1\}^n} p_{y \oplus z}(f \cdot \chi_y) \cdot \hat{f}(z)^2\\
&= \frac{1}{N2^N} \sum_{f \in \Fn} \sum_{z \in \{0,1\}^n} \sum_{y \in \{0,1\}^n} p_{y \oplus z}(f \cdot \chi_y) \cdot \widehat{f \cdot \chi_y}(y \oplus z)^2\\
&= \frac{1}{N2^N} \sum_{y \in \{0,1\}^n} \sum_{f \in \Fn} \sum_{z \in \{0,1\}^n} p_{y \oplus z}(f \cdot \chi_y) \cdot \widehat{f \cdot \chi_y}(y \oplus z)^2\\
&= \frac{1}{2^N} \sum_{f \in \Fn} \sum_{z \in \{0,1\}^n} p_z(f) \cdot \hat{f}(z)^2.
\end{align*}
The nonnegativity constraint on each $p'_z(f)$ is satisfied by convexity, and the polynomials sum to $1$ for each $f$ because
\begin{align*}
\sum_{z \in \{0,1\}^n} p'_z(f) &= \frac{1}{N} \sum_{z \in \{0,1\}^n} \sum_{y \in \{0,1\}^n} p_{y \oplus z}(f \cdot \chi_y)\\
&= \frac{1}{N} \sum_{y \in \{0,1\}^n} \sum_{z \in \{0,1\}^n} p_{y \oplus z}(f \cdot \chi_y)\\
&= \frac{1}{N} \sum_{y \in \{0,1\}^n} \sum_{z \in \{0,1\}^n} p_z(f \cdot \chi_y)\\
&= \frac{1}{N} \sum_{y \in \{0,1\}^n} 1\\
&= 1.
\end{align*}
Notice that the $p'_z$s satisfy $p'_z(f) = p'_{0^n} (f \cdot \chi_z)$. So, we can rewrite the linear program in terms of $p'_{0^n}(f)$ alone. Define $p(f) = p'_{0^n}(f)$ and define the coefficients of $p(f)$ by:
$$p(f) = \sum_{S \subset \{0,1\}^n, |S| \le 2T} c_S \cdot \prod_{x \in S} f(x).$$
Now, we can rewrite the linear program \eqref{eq:primal_lp_unsimplified} in terms of $p(f)$ as:
\begin{equation}
\label{eq:primal_lp_symmetrized}
\boxed{\begin{array}{lll} 
    \text{max}  & \displaystyle\frac{1}{2^N} \sum_{f \in \Fn} \sum_{z \in \{0,1\}^n} p(f \cdot \chi_z) \cdot \hat{f}(z)^2    \\
    \text{subject to} & p(f) \ge 0 & \text{ for each } f \in \Fn\\
    &\displaystyle\sum_{z \in \{0,1\}^n} p(f \cdot \chi_z) = 1 & \text{ for each } f \in \Fn\\
    &c_S \in \mathbb{R} & \text{ for each } 0 \le |S| \le 2T
\end{array}}
\end{equation}

We can simplify the objective function of the linear program \eqref{eq:primal_lp_symmetrized} even further:
\begin{align*}
\frac{1}{2^N} \sum_{f \in \Fn} \sum_{z \in \{0,1\}^n} p(f \cdot \chi_z) \cdot \hat{f}(z)^2 &= \frac{1}{2^N} \sum_{f \in \Fn} \sum_{z \in \{0,1\}^n} p(f \cdot \chi_z) \cdot \widehat{f \cdot \chi_z}(0^n)^2\\
&= \frac{1}{2^N} \sum_{z \in \{0,1\}^n} \sum_{f \in \Fn} p(f \cdot \chi_z) \cdot \widehat{f \cdot \chi_z}(0^n)^2\\
&= \frac{1}{2^N} \sum_{z \in \{0,1\}^n} \sum_{f \in \Fn} p(f) \cdot \widehat{f}(0^n)^2\\
&= \frac{N}{2^N} \sum_{f \in \Fn} p(f) \cdot \hat{f}(0^n)^2.
\end{align*}
Notice that we can also assume $p(f) = p(-f)$ without loss of generality, because the squared Fourier coefficient of $f$ is the same as the squared Fourier coefficient of its negation, and because replacing $p(f)$ by $\frac{p(f) + p(-f)}{2}$ still satisfies all of the constraints. In particular,we can assume $c_S = 0$ if $|S|$ is odd.

Next, we turn to simplifying the equality constraint. Define $q(f)$ by:
\begin{align*}
q(f) &:= \sum_{z \in \{0,1\}^n} p(f \cdot \chi_z)\\
&= \sum_{z \in \{0,1\}^n} \sum_{|S| \le 2T} c_S \cdot \prod_{x \in S} f(x) \cdot (-1)^{x \cdot z},
\end{align*}
which is also a multilinear polynomial in $f$ of degree $2T$. Then the equality constraint reads as $q(f) = 1$ for every $f \in \Fn$. Because $q$ is multilinear, this implies $q(f) = 1$ in fact holds identically over all $f : \{0,1\}^n \to \Reals$, and not just Boolean-valued $f$. So, the coefficient on the monomial of the set $S$ in $q$ must be $1$ if $S = \emptyset$ and $0$ otherwise. For $S$ empty, we have:
$$\sum_{z \in \{0,1\}^n} c_\emptyset = 1,$$
which is to say that $c_\emptyset = \frac{1}{N}$. Otherwise, for nonempty $S$ we have:
$$\sum_{z \in \{0,1\}^n} c_S \prod_{x \in S}(-1)^{x \cdot z} = 0.$$
We can rewrite the equation above as:
$$\sum_{z \in \{0,1\}^n} c_S (-1)^{\left(\sum_{x \in S} x \right) \cdot z} = 0.$$
Now, there are two cases:
\begin{itemize}
\item If $\mathop{\bigoplus}_{x \in S} x = 0^n$, then the equation holds if and only if $c_S = 0$.
\item If $\mathop{\bigoplus}_{x \in S} x \neq 0^n$, then the terms in the sum are $c_S$ half of the time and $-c_S$ the other half of the time, so equality always holds.
\end{itemize}
Putting this altogether, we can conclude:
\begin{itemize}
\item $c_{\emptyset} = \frac{1}{N}$.
\item $c_S = 0$ if $\mathop{\bigoplus}_{x \in S} x = 0^n$.
\item $c_S = 0$ if $|S|$ is odd.
\item $c_S$ is otherwise unconstrained by $q(f) = 1$.
\end{itemize}

After all of this, our linear program now has the much simpler form:\footnote{Here, $\bigoplus S$ is shorthand for $\bigoplus_{x \in S} x$.}
\begin{equation}
\label{eq:primal_lp_variables_reduced}
\boxed{\begin{array}{lll} 
    \text{max}  & \displaystyle\frac{N}{2^N} \sum_{f \in \Fn} p(f) \cdot \hat{f}(0^n)^2    \\
    \text{subject to} & p(f) \ge 0 & \text{ for each } f \in \Fn\\
    &\displaystyle c_\emptyset = \frac{1}{N}\\
    &c_S \in \mathbb{R} & \text{ for each } 2 \le |S| \le 2T \text{ with } |S| \text{ even, }\bigoplus S \neq 0^n
    \end{array}}
\end{equation}
    
Alternatively, we can express the linear program \eqref{eq:primal_lp_variables_reduced} purely in terms of the variables $c_S$, rather than leaving them implicit in $p(f)$. In the objective function, the coefficient on $c_S$ is given by:
$$k_S := \frac{N}{2^N} \sum_{f \in \Fn} \hat{f}(0^n)^2 \cdot \prod_{x \in S} f(x).$$
We compute $k_S$ depending on the size of $S$:
\begin{itemize}
\item If $S = \emptyset$, then $k_S = N \cdot \E\left[\hat{f}(0^n)^2\right] = N \cdot \E\left[\hat{f}(0^n)^2 - \E\left[\hat{f}(0^n)\right]^2\right] = N \cdot \Var\left[\hat{f}(0^n)\right] = 1$, because for a random $f$, $\hat{f}(0^n)$ is a sum of $2^n$ independent $\pm \frac{1}{2^n}$ variables.
\item If $S \neq \emptyset$, then:
\begin{align*}
k_S &= \frac{N}{2^N} \sum_{f \in \Fn} \frac{1}{2^{2n}} \sum_{x_1 \in \{0,1\}^n} \sum_{x_2 \in \{0,1\}^n} f(x_1)f(x_2)\prod_{x \in S}f(x)\\
&= \frac{N}{2^N} \frac{1}{2^{2n}} \sum_{x_1 \in \{0,1\}^n} \sum_{x_2 \in \{0,1\}^n} \sum_{f \in \Fn} f(x_1)f(x_2)\prod_{x \in S}f(x)\\
&= \begin{cases}\frac{2}{N} & |S| = 2\\0 & |S| > 2 \end{cases},
\end{align*}
because in the second line, the innermost sum is $0$ unless $\{x_1, x_2\} = S$.
\end{itemize}
So, the final primal linear program takes the form:

\begin{equation}
\label{eq:primal_lp_final}
\boxed{\begin{array}{lll} 
    \text{max}  & \displaystyle c_\emptyset + \frac{2}{N}\sum_{|S| = 2} c_S  \\
    \text{subject to} & \displaystyle \sum_S c_S \cdot \prod_{x \in S} f(x) \ge 0 & \text{ for each } f \in \Fn\\
    &\displaystyle c_\emptyset = \frac{1}{N}\\
    &c_S \in \mathbb{R} & \text{ for each } 2 \le |S| \le 2T \text{ with } |S| \text{ even, }\bigoplus S \neq 0^n
    \end{array}}
\end{equation}

Standard manipulations reveal the dual linear program of \eqref{eq:primal_lp_final}:

\begin{equation}
\label{eq:dual_lp}
\boxed{\begin{array}{lll} 
    \text{min}  &     \displaystyle \frac{b}{N}\\
    \text{subject to} & \displaystyle b - \sum_{f \in \Fn} \psi_f = 1\\
    & \displaystyle - \sum_{f \in \Fn} \psi_f\prod_{x \in S}f(x) = \frac{2}{N} & \text{ for each } |S| = 2\\
    & \displaystyle - \sum_{f \in \Fn} \psi_f\prod_{x \in S}f(x) = 0 & \text{ for each } 4 \le |S| \le 2T \text{ with } |S| \text{ even, }\bigoplus S \neq 0^n\\
    &\psi_f \ge 0 & \text{ for each } f \in \Fn\\
    &b \in \mathbb{R}
    \end{array}}
\end{equation}
    
We now construct a solution to the dual linear program \eqref{eq:dual_lp} for $T=1$ query. Our dual solution is motivated by complementary slackness, which guarantees that a variable in \eqref{eq:dual_lp} of the optimal dual solution is nonzero if and only if the corresponding constraint in \eqref{eq:primal_lp_final} is tight in the optimal primal solution. The naive XHOG algorithm solves the primal linear program with $p(f) = \hat{f}(0^n)^2$, so $p(f) = 0$ if and only if $\hat{f}(0^n) = 0$. Thus, if we think that the naive algorithm is optimal, then we should look for a dual solution where $\psi_f$ is nonzero if and only if $\hat{f}(0^n) = 0$.

For some $\kappa$ to be chosen later, we choose $\psi_f = \kappa$ if $\hat{f}(0^n) = 0$ and $\psi_f = 0$ otherwise. In other words, we let $\psi_f = \kappa \cdot \Half_N(f)$, where $\Half_N : \{-1,1\}^N \to \{0,1\}$ is the $0$-$1$ indicator of the set of functions in $\Fn$ (viewed as $N$-bit strings) with exactly $\frac{N}{2}$ coordinates equal to $-1$.

Viewing $\psi_f = \psi(f)$ as a function $\psi: \{-1,1\}^N \to \Reals$, it will be convenient to rewrite the constraints of the dual linear program in terms of the Fourier coefficients of $\psi$. We use the inner product formula below for Fourier coefficients, with the understanding that we identify a set $S \subseteq [N]$ with its characteristic string in $\{0,1\}^N$:
$$\hat{\psi}(S) = \frac{1}{2^N} \sum_{f \in \{-1,1\}^N} \psi_f \prod_{x \in S}f(x).$$
Now, the dual linear program reads as:
\begin{equation}
\label{eq:dual_lp_fourier}
\boxed{\begin{array}{lll} 
    \text{min}  &     \displaystyle \frac{b}{N}\\
    \text{subject to} & \displaystyle b - 2^N\hat{\psi}(\emptyset) = 1\\
    & \displaystyle 2^N\hat{\psi}(S) = -\frac{2}{N} & \text{ for each } |S| = 2\\
    & \displaystyle \hat{\psi}(S) = 0 & \text{ for each } 4 \le |S| \le 2T \text{ with } |S| \text{ even, }\bigoplus S \neq 0^n\\
    &\psi_f \ge 0 & \text{ for each } f \in \Fn\\
    &b \in \mathbb{R}
    \end{array}}
\end{equation}

The Fourier coefficients of $\Half_N$ are well known \cite[Theorem~5.19]{OD14}, though they are also easy to compute by hand for sets of small size. For $|S| = 2j$, they are given by:
$$\widehat{\Half_N}(S) = (-1)^j \frac{\binom{N/2}{j}}{\binom{N}{2j}} \cdot \frac{\binom{N}{N/2}}{2^N}.$$
The $|S| = 2$ equality constraint of the dual linear program \eqref{eq:dual_lp_fourier} implies
$$2^N \cdot \kappa \cdot \widehat{\Half_N}(S) = -\frac{2}{N}$$
$$\kappa = \frac{1}{2^N} \cdot \frac{2}{N} \cdot \frac{\binom{N}{2}}{\binom{N/2}{1}} \cdot \frac{2^N}{\binom{N}{N/2}} = \frac{4\binom{N}{2}}{N^2\binom{N}{N/2}} = \frac{2(N-1)}{N\binom{N}{N/2}}.$$
Plugging this value of $\kappa$ into the constraint on $\hat{\psi}(\emptyset)$ gives:
$$b - 2^N \cdot \kappa \cdot \widehat{\Half_N}(\emptyset) = 1$$
$$b = 1 + 2^N \cdot \kappa \cdot \frac{\binom{N/2}{0}}{\binom{N}{0}} \cdot \frac{\binom{N}{N/2}}{2^N} = 1 + 2\frac{N-1}{N} = 3 - \frac{2}{N}.$$
This completes the proof, as we have shown a solution to the dual linear program with objective value $\frac{3 - \frac{2}{N}}{N}$.

\section{Discussion}
\label{sec:discussion}
The most natural question left for future work is whether our bounds could be improved. Our lower bounds for $b$-XHOG with $\Oracle_\psi$ or $U_\psi$ show that for constant $\eps$, $(2 + \eps)$-XHOG requires $\Omega\left(\frac{2^{n/4}}{\poly(n)}\right)$ queries to either oracle, while the best upper bound we know of solves $(2 + \eps)$-XHOG in $O\left(2^{n/3}\right)$ queries. We conjecture that this upper bound is tight.

One possible approach towards improving the lower bound for $b$-XHOG with $\Oracle_\psi$ (and by extension, $U_\psi$) is to use the polynomial method, as we did for the $\FS$ lower bound. Indeed, the output probabilities of an algorithm that makes $T$ queries to $\Oracle_\psi$ can be expressed as degree-$2T$ polynomials in the entries of $\Oracle_\psi$. If we write $\ket{\psi} = \sum_{i=1}^{N} \alpha_i \ket{i}$, then these are polynomials in the amplitudes $\alpha_1,\ldots,\alpha_N$ and the conjugates of the amplitudes $\alpha_1^*,\ldots,\alpha_N^*$. Because of the invariance of the Haar measure with respect to phases, and because the linear XEB score depends only on the magnitudes of the amplitudes, we can further assume without loss of generality that the output probabilities are polynomials in the variables $|\alpha_1|^2, \ldots, |\alpha_N|^2$, which are equivalently the measurement probabilities of $\ket{\psi}$ in the computational basis. We can also assume that these polynomials are homogeneous, because the input variables satisfy $\sum_{i=1}^N |\alpha_i|^2 = 1$, so we can multiply any lower-degree terms by this sum to make all terms have the same degree. Like in our $\FS$ lower bound, the polynomials are constrained to represent a probability distribution for all valid inputs. However, unlike the $\FS$ lower bound, this introduces uncountably many constraints in the primal linear program: the polynomials representing the output probabilities must be nonnegative for \textit{all} inputs $(|\alpha_1|^2, \ldots, |\alpha_N|^2)$ on the probability simplex. It may still be possible to exhibit a solution to the dual linear program if only finitely many of the constraints are relevant (such an approach was used in \cite{BT13}, for example). Put another way, it might suffice to throw out all but finitely many of the primal constraints to obtain a nontrivial upper bound on the value of the linear program.

Our $b$-XHOG bound for $\FS$ circuits is tight, but it only applies to single-query algorithms. In principle, our lower bound approach via the polynomial method could be generalized to algorithms that make additional queries, by increasing the degree of the polynomials in the linear program \eqref{eq:primal_lp_final} and exhibiting another dual solution. The challenge seems to be that the parity constraint on the monomials with nonzero coefficients becomes unwieldy when working with polynomials of larger degree.

Beyond possible improvements to the query complexity bounds, it would be interesting to give some evidence that beating the naive XHOG algorithm is hard in the real world. Aaronson and Gunn \cite{AG19} showed that $(1 + \eps)$-XHOG is \textit{classically} hard, assuming the classical hardness of nontrivially estimating the output probabilities of random quantum circuits. It is not clear whether a similar argument could work for quantum algorithms, though, because sampling from a random quantum circuit gives a better-than-trivial algorithm for estimating its output probabilities.



\section*{Acknowledgements}
Thanks to Scott Aaronson, Sabee Grewal, Sam Gunn, Robin Kothari, Daniel Liang, Patrick Rall, Andrea Rocchetto, and Justin Thaler for helpful discussions and illuminating insights. Thanks also to anonymous reviewers for helpful comments regarding the presentation of this work. This work was supported by a Vannevar Bush Fellowship and a National Defense Science and Engineering Graduate (NDSEG) Fellowship from the US Department of Defense.


\phantomsection\addcontentsline{toc}{section}{References}
\bibliographystyle{plainnat}
\bibliography{Tsirelson}

\end{document}